\documentclass[11pt,a4paper]{article}
\pdfoutput=1

\usepackage[utf8]{inputenc}
\usepackage[margin=1.1in]{geometry}
\usepackage{amsmath,amssymb,amsthm,mathtools}
\usepackage{microtype}
\usepackage{booktabs}
\usepackage{array}
\usepackage[hidelinks]{hyperref}
\usepackage[capitalise,nameinlink]{cleveref}

\hypersetup{
  pdftitle={Removing Trace Preservation from a Universal Spectral Bound for 2-Positive Maps},
  pdfauthor={Petri Tanninen},
  pdfsubject={Mathematical physics; spectral theory of positive maps},
  pdfkeywords={2-positive maps, conditionally completely positive maps, quantum dynamical semigroups, relaxation rates, eigenvalue inequalities}
}

\theoremstyle{plain}
\newtheorem{theorem}{Theorem}
\newtheorem{lemma}[theorem]{Lemma}
\newtheorem{proposition}[theorem]{Proposition}
\newtheorem{corollary}[theorem]{Corollary}
\newtheorem*{theoremA}{Theorem A {\normalfont(vom Ende, Chru\'sci\'nski, Kimura, Muratore-Ginanneschi~\cite{vEC+26})}}
\newtheorem*{conjecture}{Conjecture {\normalfont(\cite[Sec.~3, Eq.~(10)]{vEC+26})}}
\theoremstyle{definition}
\newtheorem{definition}[theorem]{Definition}
\theoremstyle{remark}
\newtheorem{remark}[theorem]{Remark}

\crefname{theorem}{Theorem}{Theorems}
\crefname{lemma}{Lemma}{Lemmas}
\crefname{proposition}{Proposition}{Propositions}
\crefname{corollary}{Corollary}{Corollaries}
\crefname{remark}{Remark}{Remarks}
\crefname{definition}{Definition}{Definitions}

\numberwithin{equation}{section}

\newcommand{\C}{\mathbb{C}}
\newcommand{\R}{\mathbb{R}}
\newcommand{\M}{M_d}
\newcommand{\tr}{\operatorname{tr}}
\newcommand{\Tr}{\operatorname{Tr}}
\newcommand{\Rea}{\operatorname{Re}}
\newcommand{\id}{\operatorname{id}}
\newcommand{\spec}{\sigma}
\newcommand{\mmin}{m}
\newcommand{\smax}{s}
\newcommand{\cA}{\mathcal{A}}
\newcommand{\cT}{\mathcal{T}}
\newcommand{\cU}{\mathcal{U}}
\newcommand{\cC}{\mathcal{C}}
\newcommand{\cD}{\mathcal{D}}
\newcommand{\norm}[1]{\lVert #1\rVert}
\newcommand{\ket}[1]{\lvert #1\rangle}

\newcommand{\ketbra}[2]{\lvert #1\rangle\langle #2\rvert}
\newcommand{\eps}{\varepsilon}

\title{\textbf{Removing Trace Preservation from a Universal\\ Spectral Bound for 2-Positive Maps}\thanks{Preprint archived at \href{https://doi.org/10.5281/zenodo.23225994}{doi:10.5281/zenodo.23225994}.}}
\author{%
	Petri Tanninen\\[2pt]
	\small Sentient Machine Corporation / OCAD University\\
	\small \href{mailto:petri@sentientmachine.com}{\texttt{petri@sentientmachine.com}}
}
\date{October 2026}

\begin{document}

\maketitle

\begin{abstract}
Vom Ende, Chru\'sci\'nski, Kimura and Muratore-Ginanneschi proved that every
2-positive trace-preserving map $\Phi$ on $\M=M_d(\C)$ satisfies
$\Tr\Phi\le d\,\mmin(\Phi)+d^2-d$, where $\mmin(\Phi)$ is the smallest real part of
an eigenvalue of $\Phi$, and deduced the universal relaxation-rate bound
$\Tr L\le d\,\mmin(L)$ for generators of 2-positive trace-preserving semigroups.
They asked whether trace preservation is needed, and conjectured that every
conditionally 2-positive map $L$, i.e.\ every $L$ for which $e^{tL}$ is
2-positive for all $t\ge0$, satisfies
\[
  \Tr L\le d\,\mmin(L)+(d^2-d)\,\smax(L),
\]
where $\smax(L)$ is the largest real part of an eigenvalue of $L$.
We prove this conjecture. The argument normalizes a 2-positive map by a
positive definite eigenmatrix of its Hilbert--Schmidt adjoint, which turns it
into a 2-positive trace-preserving map in the manner of a quantum Doob
transform. The existence of such an eigenmatrix is guaranteed after adding an
arbitrarily small multiple of the unnormalized depolarizing map, and the
inequality then passes to the limit by continuity of the spectrum. No
irreducibility, diagonalizability or faithful-invariant-state assumption is
required. We also show that the constant $d$ is optimal within the class of
non-trace-preserving completely positive semigroups, and we recast the result as a
universal constraint on decay rates measured relative to the dominant spectral
growth rate.
\end{abstract}



\section{Introduction}\label{sec:intro}

The eigenvalues of the generator of a Markovian quantum evolution determine its
relaxation rates, which are the experimentally accessible rates of decoherence,
dissipation and thermalization~\cite{BP02,AL07,RH12}. Complete positivity
constrains these rates in a nontrivial way. For a generator $L$ of a
completely positive trace-preserving (CPTP) semigroup on a $d$-level system, the
relaxation rates $\Gamma_j=-\Rea\lambda_j$ obey the universal bound
\begin{equation}\label{eq:rates-classical}
  \Gamma_k\le\frac1d\sum_{j}\Gamma_j \qquad\text{for every }k.
\end{equation}
This bound was conjectured in~\cite{CKKS21,CFKO21}, building on the qubit case
of~\cite{Kim02}, and proved in~\cite{MKC25} using Lyapunov-theoretic methods.

Vom Ende, Chru\'sci\'nski, Kimura and Muratore-Ginanneschi~\cite{vEC+26} later
gave an entirely algebraic proof that applies under the weaker hypothesis of
2-positivity. The key step in their argument is an inequality for maps rather than
generators (Theorem~A below): every 2-positive trace-preserving map $\Phi$
on $\M$ satisfies
\begin{equation}\label{eq:vEC-intro}
  \Tr\Phi\le d\,\mmin(\Phi)+d^2-d .
\end{equation}
Here $\mmin(\Phi)$ is the smallest real part of an eigenvalue of $\Phi$.
Inequality~\eqref{eq:rates-classical} then follows by differentiation along the
semigroup; see also~\cite{CvEKM25} for a direct proof at the generator level.

In the outlook of~\cite{vEC+26}, the authors note that trace preservation
enters the inequality mainly by fixing the dominant eigenvalue. That eigenvalue
equals $1$ for trace-preserving maps and $0$ for their generators. They ask
whether this hypothesis can be dropped once the dominant eigenvalue is
reinstated explicitly, and they formulate the following conjecture. Writing
$\smax(\cA)$ for the largest real part of an eigenvalue of $\cA$, and
calling a linear map $L$ on $\M$ \emph{conditionally 2-positive} if $e^{tL}$ is
2-positive for all $t\ge0$, the conjecture reads as follows.

\begin{conjecture}
Every conditionally 2-positive map $L\colon\M\to\M$ satisfies
\[
  \Tr L\le d\,\mmin(L)+(d^2-d)\,\smax(L).
\]
\end{conjecture}

The class of conditionally 2-positive maps contains every 2-positive map, every
completely positive map, and every conditionally completely positive map in the
sense of Evans~\cite{Eva77,WECC08}. The authors of~\cite{vEC+26} report numerical evidence in the
completely positive case. They also note that a proof would likely require a
substantially different approach, because their argument relies on the
existence of a fixed point, which is no longer available without trace
preservation.

\paragraph{Main result.}
We prove the conjecture.

\begin{theorem}\label{thm:main}
Let $L\colon\M\to\M$ be conditionally 2-positive. Then
\begin{equation}\label{eq:main}
  \Tr L\le d\,\mmin(L)+(d^2-d)\,\smax(L).
\end{equation}
Moreover, for each $d\ge2$, equality holds in~\eqref{eq:main} for a family of
generators of completely positive semigroups that are not trace-preserving.
In particular, in an inequality of the form
$\Tr L\le a\,\mmin(L)+(d^2-a)\,\smax(L)$, the coefficient $a=d$ cannot be
replaced by any larger number.
\end{theorem}

The proof proceeds through the corresponding statement for maps.

\begin{theorem}\label{thm:maps}
Every 2-positive map $\Phi\colon\M\to\M$ satisfies
\begin{equation}\label{eq:maps}
  \Tr\Phi\le d\,\mmin(\Phi)+(d^2-d)\,\smax(\Phi).
\end{equation}
\end{theorem}

The difficulty identified in~\cite{vEC+26} is real. A general 2-positive map has
no distinguished fixed point, its dominant eigenvalue need not be $1$, and it
may be reducible or nilpotent. Our approach is to reduce to Theorem~A instead of
reproving it. If the Hilbert--Schmidt adjoint $\Phi^*$ has a \emph{positive
definite} eigenmatrix $H$ with eigenvalue $r>0$, then conjugation by $H^{1/2}$
followed by division by $r$ produces a 2-positive \emph{trace-preserving} map
with the same spectrum up to the factor $r$. This normalization is a version of
the quantum Doob transform used in the large-deviation theory of open quantum
systems~\cite{CGLP18}. A positive definite eigenmatrix need not exist, but it does
exist after adding $\eps\,\tr(\cdot)\,I$ to $\Phi$, by Brouwer's fixed-point
theorem. The inequality is stable as $\eps\downarrow0$ because it involves only
the trace and the extreme real parts of the spectrum, all of which depend
continuously on the map. In particular, the normalizing matrices $H_\eps$ do not
need to converge, so degenerate limits cause no difficulty.

Theorem~\ref{thm:main} has a natural interpretation in terms of decay rates. Set
$\Gamma_j=\smax(L)-\Rea\lambda_j\ge0$, the decay rates measured relative to the
dominant growth rate. Then~\eqref{eq:main} is equivalent to
$\max_k\Gamma_k\le\frac1d\sum_j\Gamma_j$
(\cref{cor:rates}), which extends~\eqref{eq:rates-classical} verbatim to
non-trace-preserving dynamics. Such dynamics arise, for instance, as tilted
generators in the study of quantum-trajectory fluctuations~\cite{CGLP18}.

\paragraph{Organization.}
\Cref{sec:prelim} fixes notation and records the input from~\cite{vEC+26}.
\Cref{sec:spectral} contains a standard spectral lemma for positive
trace-preserving maps. \Cref{sec:normalization} proves \cref{thm:maps} in the
presence of a positive definite eigenmatrix, and \cref{sec:general} removes that
assumption. \Cref{sec:generators} deduces \cref{thm:main}. \Cref{sec:sharp}
establishes sharpness, and \cref{sec:rates} discusses the decay-rate formulation.

\section{Notation and preliminaries}\label{sec:prelim}

Throughout, $d$ is a positive integer and all maps are complex-linear.
Let $\M=M_d(\C)$ with the Hilbert--Schmidt inner product
$\langle A,B\rangle=\tr(A^\dagger B)$. For a linear map $\cA\colon\M\to\M$,
$\cA^*$ denotes its Hilbert--Schmidt adjoint. We write $\tr$ for the trace of a
$d\times d$ matrix and $\Tr\cA$ for the trace of $\cA$ as an operator on the
$d^2$-dimensional space $\M$. Thus $\Tr\cA=\sum_{\lambda}\lambda$, with eigenvalues
counted according to algebraic multiplicity. The spectrum of $\cA$ is
$\spec(\cA)$, and we set
\[
  \mmin(\cA)=\min_{\lambda\in\spec(\cA)}\Rea\lambda,
  \qquad
  \smax(\cA)=\max_{\lambda\in\spec(\cA)}\Rea\lambda .
\]
For matrices, $X\ge0$ means positive semidefinite and $X>0$ means positive definite.
The operator norm is $\norm{\cdot}_\infty$.

\begin{definition}
A linear map $\Phi\colon\M\to\M$ is \emph{positive} if $\Phi(X)\ge0$ whenever
$X\ge0$. It is \emph{$k$-positive} if $\id_{M_k}\otimes\Phi$ is positive on
$M_k\otimes\M$, and \emph{completely positive} if it is $k$-positive for every
$k$. It is \emph{trace-preserving} if $\tr\Phi(X)=\tr X$ for all $X$. A linear map
$L\colon\M\to\M$ is \emph{conditionally 2-positive} if $e^{tL}$ is 2-positive for
every $t\ge0$.
\end{definition}

We record some elementary facts that will be used repeatedly.

\begin{enumerate}
\renewcommand{\theenumi}{(P\arabic{enumi})}\renewcommand{\labelenumi}{\theenumi}
\item \label{P1}
  \emph{Hermiticity and real traces.} A positive map is Hermitian-preserving,
  since every Hermitian $X$ is a difference of two positive semidefinite matrices.
  If $e^{tL}$ is Hermitian-preserving for all $t\ge0$, then so is
  $L=\lim_{t\downarrow0}t^{-1}(e^{tL}-\id)$. A Hermitian-preserving map is
  represented by a real matrix in any basis of Hermitian matrices. Hence
  its characteristic polynomial has real coefficients, its spectrum is
  invariant under complex conjugation, and its trace is real. In
  particular, all quantities appearing in~\eqref{eq:main} and~\eqref{eq:maps}
  are real.
\item \label{P2}
  \emph{Positivity of the adjoint.} If $\Phi$ is positive, then so is $\Phi^*$.
  Indeed, for $Y\ge0$ and $v\in\C^d$ we have
  $\langle v,\Phi^*(Y)v\rangle=\tr\bigl(Y\,\Phi(vv^\dagger)\bigr)\ge0$, because the
  trace of a product of two positive semidefinite matrices is nonnegative.
\item \label{P3}
  \emph{Congruences.} For any $S\in\M$, the map $X\mapsto SXS^\dagger$ is completely
  positive, since its $k$-fold amplification is congruence by $I_k\otimes S$.
  The 2-positive maps form a convex cone that is closed under composition and
  topologically closed.
\item \label{P4}
  \emph{Continuity of the spectrum.} The eigenvalues of a matrix, counted with
  algebraic multiplicity, depend continuously on its entries as an unordered
  $d^2$-tuple, because the roots of a monic polynomial depend continuously on its
  coefficients~\cite[Ch.~II, Thm.~5.14]{Kat80}. Consequently, $\cA\mapsto\mmin(\cA)$,
  $\cA\mapsto\smax(\cA)$ and $\cA\mapsto\Tr\cA$ are continuous.
\end{enumerate}

The only specialized result we use is the main theorem of~\cite{vEC+26}.

\begin{theoremA}[{\cite[Thm.~1]{vEC+26}}]
If $\cT\colon\M\to\M$ is 2-positive and trace-preserving, then
\begin{equation}\label{eq:thmA}
  \Tr\cT\le d\,\mmin(\cT)+d^2-d .
\end{equation}
\end{theoremA}

\section{A spectral lemma for positive trace-preserving maps}\label{sec:spectral}

The following fact is well known; see, e.g., \cite[Prop.~4.26]{Wat18}. We include a
short proof for completeness.

\begin{lemma}\label[lemma]{lem:spectrum}
Let $\cT\colon\M\to\M$ be positive and trace-preserving. Then $1\in\spec(\cT)$ and
$\lvert\lambda\rvert\le1$ for every $\lambda\in\spec(\cT)$. In particular,
$\smax(\cT)=1$.
\end{lemma}

\begin{proof}
Let $\cU=\cT^*$. By~\ref{P2}, $\cU$ is positive. Trace preservation of $\cT$ is
equivalent to $\cU(I)=I$, so $1\in\spec(\cU)$. Since $\spec(\cT)=\overline{\spec(\cU)}$,
it follows that $1\in\spec(\cT)$.

For Hermitian $A$ we have $-\norm{A}_\infty I\le A\le\norm{A}_\infty I$. Since each
power $\cU^n$ is positive and unital, this gives
$\norm{\cU^n(A)}_\infty\le\norm{A}_\infty$. Now write an arbitrary $X\in\M$ as $X=A+iB$ with
$A=\tfrac12(X+X^\dagger)$ and $B=\tfrac1{2i}(X-X^\dagger)$ Hermitian. Then
$\norm{A}_\infty,\norm{B}_\infty\le\norm{X}_\infty$, and therefore
$\norm{\cU^n(X)}_\infty\le2\norm{X}_\infty$. If $\cU(X)=\lambda X$ with $X\ne0$, it
follows that $\lvert\lambda\rvert^n\le2$ for all $n\in\mathbb{N}$, so
$\lvert\lambda\rvert\le1$. Passing to complex conjugates proves the claim for $\cT$.
\end{proof}

\section{Normalization by a positive definite eigenmatrix}\label{sec:normalization}

\begin{lemma}\label[lemma]{lem:normalization}
Let $\Phi\colon\M\to\M$ be 2-positive, and suppose that there are a positive
definite $H\in\M$ and a real number $r>0$ such that
\begin{equation}\label{eq:eigenmatrix}
  \Phi^*(H)=rH .
\end{equation}
Then $\smax(\Phi)=r$, and $\Phi$ satisfies~\eqref{eq:maps}.
\end{lemma}

\begin{proof}
Let $\cC_H(X)=H^{1/2}XH^{1/2}$. This map is invertible, and its inverse is
$\cC_H^{-1}(X)=H^{-1/2}XH^{-1/2}$. Define
\begin{equation}\label{eq:doob}
  \cT=\frac1r\,\cC_H\circ\Phi\circ\cC_H^{-1}.
\end{equation}
By~\ref{P3}, $\cC_H$ and $\cC_H^{-1}$ are completely positive. Since $r>0$, $\cT$ is
2-positive. It is also trace-preserving: for every $X\in\M$, using cyclicity of
the trace, the Hermiticity of $H$ and~\eqref{eq:eigenmatrix},
\begin{align*}
  \tr\cT(X)
  &=\frac1r\tr\Bigl[H\,\Phi\bigl(H^{-1/2}XH^{-1/2}\bigr)\Bigr]
   =\frac1r\tr\Bigl[\Phi^*(H)\,H^{-1/2}XH^{-1/2}\Bigr]\\
  &=\tr\Bigl[H\,H^{-1/2}XH^{-1/2}\Bigr]
   =\tr X .
\end{align*}
Because~\eqref{eq:doob} is a similarity transformation followed by multiplication
by $r^{-1}$, we have
\begin{equation}\label{eq:scaling}
  \spec(\cT)=r^{-1}\spec(\Phi),\qquad
  \Tr\cT=r^{-1}\Tr\Phi,\qquad
  \mmin(\cT)=r^{-1}\mmin(\Phi),
\end{equation}
with algebraic multiplicities preserved. By \cref{lem:spectrum}, $1\in\spec(\cT)$ and
every eigenvalue of $\cT$ has modulus at most $1$. Hence $r\in\spec(\Phi)$ and
every eigenvalue of $\Phi$ has modulus at most $r$, so $\smax(\Phi)=r$.
Applying Theorem~A to $\cT$ and using~\eqref{eq:scaling},
\[
  \frac{\Tr\Phi}{r}\le d\,\frac{\mmin(\Phi)}{r}+d^2-d .
\]
Multiplying by $r=\smax(\Phi)>0$ gives~\eqref{eq:maps}.
\end{proof}

\begin{remark}
The map~\eqref{eq:doob} is a map-level analogue of the quantum Doob transform
of~\cite{CGLP18}, which there normalizes tilted generators using the dominant
eigenmatrix. In~\cite{vEC+26} a related $\omega$-weighted inner product is used to
symmetrize trace-preserving maps that have a faithful fixed point. Here we use
the normalization in the opposite direction, to \emph{produce} trace
preservation.
\end{remark}

\section{Removing the eigenmatrix hypothesis: proof of \texorpdfstring{\cref{thm:maps}}{Theorem 2}}\label{sec:general}

A general 2-positive map need not admit a positive definite eigenmatrix
for its adjoint. For example, the adjoint may be nilpotent, or every
positive semidefinite eigenmatrix may be singular because of reducibility.
We remove this hypothesis by perturbation.

\begin{proof}[Proof of \cref{thm:maps}]
Let $\Phi$ be 2-positive and let
\[
  \Delta(X)=\tr(X)\,I=\sum_{i,j=1}^d E_{ij}XE_{ij}^\dagger ,
\]
where $E_{ij}=\ketbra{i}{j}$ are the matrix units. By~\ref{P3}, $\Delta$ is completely
positive, and it is self-adjoint with respect to the Hilbert--Schmidt inner product.
For $\eps>0$ put
\[
  \Phi_\eps=\Phi+\eps\Delta ,
\]
which is 2-positive by~\ref{P3}.

Let $\cD_d=\{H\in\M:\ H\ge0,\ \tr H=1\}$ be the set of density matrices. It is
compact, convex and nonempty. For $H\in\cD_d$, positivity of $\Phi^*$ (by~\ref{P2})
gives
\begin{equation}\label{eq:lowerbound}
  \Phi_\eps^*(H)=\Phi^*(H)+\eps\,I\ \ge\ \eps\,I .
\end{equation}
In particular, $\tr\Phi_\eps^*(H)\ge\eps d>0$, so
\[
  F_\eps(H)=\frac{\Phi_\eps^*(H)}{\tr\Phi_\eps^*(H)}
\]
is a continuous self-map of $\cD_d$. By Brouwer's fixed-point theorem there is
$H_\eps\in\cD_d$ with $F_\eps(H_\eps)=H_\eps$. Setting
$r_\eps=\tr\Phi_\eps^*(H_\eps)>0$, we obtain
\[
  \Phi_\eps^*(H_\eps)=r_\eps H_\eps ,
\]
and~\eqref{eq:lowerbound} yields $H_\eps\ge(\eps/r_\eps)\,I>0$. Therefore
\cref{lem:normalization} applies to $\Phi_\eps$, and
\begin{equation}\label{eq:perturbed}
  \Tr\Phi_\eps\le d\,\mmin(\Phi_\eps)+(d^2-d)\,\smax(\Phi_\eps)
  \qquad\text{for every }\eps>0 .
\end{equation}
As $\eps\downarrow0$, $\Phi_\eps\to\Phi$. By~\ref{P4}, $\Tr\Phi_\eps\to\Tr\Phi$,
$\mmin(\Phi_\eps)\to\mmin(\Phi)$ and $\smax(\Phi_\eps)\to\smax(\Phi)$. Letting
$\eps\downarrow0$ in~\eqref{eq:perturbed} proves~\eqref{eq:maps}.
\end{proof}

\begin{remark}\label[remark]{rem:limit}
The limiting argument uses only the perturbed maps $\Phi_\eps$ and their spectra.
It does not require convergence of the eigenmatrices $H_\eps$, their inverses,
or the normalized maps $\cT_\eps$, nor an a priori positive lower bound on
$r_\eps$. Singular limiting eigenmatrices, a vanishing spectral radius, and
nontrivial Jordan blocks therefore cause no difficulty. As a by-product,
$\smax(\Phi)=\lim_\eps r_\eps$ is the spectral radius of $\Phi$ and
is an eigenvalue, which recovers the Perron--Frobenius property of positive maps
(cf.~\cite{EHK78}). The argument as a whole uses 2-positivity only through
Theorem~A.
\end{remark}

\section{Generators: proof of \texorpdfstring{\cref{thm:main}}{Theorem 1}}\label{sec:generators}

\begin{proof}[Proof of inequality~\eqref{eq:main}]
Let $L$ be conditionally 2-positive, and let $\lambda_1,\dots,\lambda_{d^2}$ be its
eigenvalues counted with algebraic multiplicity. For each $t>0$ the map $e^{tL}$ is
2-positive, so \cref{thm:maps} gives
\begin{equation}\label{eq:semigroup}
  \Tr e^{tL}\le d\,\mmin(e^{tL})+(d^2-d)\,\smax(e^{tL}) .
\end{equation}
By the spectral mapping theorem (which holds via the Jordan form,
without any diagonalizability assumption), the eigenvalues of $e^{tL}$ are
$e^{t\lambda_1},\dots,e^{t\lambda_{d^2}}$, with multiplicities. Since there are
finitely many $\lambda_j$, Taylor's theorem gives a constant $C$ such that
\[
  \bigl\lvert\Rea e^{t\lambda_j}-1-t\Rea\lambda_j\bigr\rvert\le Ct^2
  \qquad\text{for all }j\text{ and all }t\in(0,1].
\]
Taking a minimum or maximum over $j$ preserves this uniform error, so
\[
  \mmin(e^{tL})=1+t\,\mmin(L)+O(t^2),\qquad
  \smax(e^{tL})=1+t\,\smax(L)+O(t^2),
\]
and furthermore $\Tr e^{tL}=\sum_je^{t\lambda_j}=d^2+t\Tr L+O(t^2)$.
Substituting these expansions into~\eqref{eq:semigroup}, the constant terms $d^2$ cancel
on both sides. Dividing by $t>0$ and letting $t\downarrow0$ yields~\eqref{eq:main}.
\end{proof}

\begin{remark}\label[remark]{rem:inclusions}
Every 2-positive map $\Phi$ is conditionally 2-positive. Indeed, each power $\Phi^n$
is 2-positive, and $e^{t\Phi}=\sum_{n\ge0}\frac{t^n}{n!}\Phi^n$ is a norm limit of
nonnegative combinations of 2-positive maps. Hence \cref{thm:main} contains
\cref{thm:maps} as well as the generator statement, which settles both readings of
the conjecture in~\cite{vEC+26}. If $e^{tL}$ is trace-preserving for every
$t\ge0$, then $L^*(I)=0$ and \cref{lem:spectrum} gives
$|e^{t\lambda}|\le1$ for each eigenvalue $\lambda$ of $L$. Thus
$0\in\spec(L)$, $\Rea\lambda\le0$, and $\smax(L)=0$. If $\Phi$ is positive
and trace-preserving, then $\smax(\Phi)=1$. In these respective cases,
\cref{thm:main,thm:maps} reduce to~\cite[Thms.~2 and~1]{vEC+26}.
\end{remark}

\section{Sharpness for non-trace-preserving dynamics}\label{sec:sharp}

The bound~\eqref{eq:main} is attained by generators of completely positive
semigroups that are \emph{not} trace-preserving, so the extension is sharp
precisely in the regime it was designed for.

\begin{proposition}\label[proposition]{prop:sharp}
Fix $d\ge2$ and an orthonormal basis $\{\ket{0},\dots,\ket{d-1}\}$ of $\C^d$. Put
$V=\ketbra{0}{1}$ and $P=\ketbra{1}{1}$, and define
\begin{equation}\label{eq:G}
  G(X)=VXV^\dagger-\tfrac12(PX+XP),\qquad
  L_c=G+c\,\id_{\M}\quad(c\in\R).
\end{equation}
Then $e^{tL_c}$ is completely positive for all $t\ge0$, and it is trace-preserving
for some $t>0$ only if $c=0$. Moreover,
\[
  \Tr L_c=cd^2-d,\qquad \mmin(L_c)=c-1,\qquad \smax(L_c)=c,
\]
so that equality holds in~\eqref{eq:main} for every $c\in\R$.
\end{proposition}

\begin{proof}
On matrix units $E_{ij}=\ketbra{i}{j}$ we have
\[
  G(E_{ij})=\delta_{i1}\delta_{j1}\,E_{00}-\tfrac12(\delta_{i1}+\delta_{j1})\,E_{ij}.
\]
This yields the eigen-decomposition in \cref{tab:spectrum}. The listed eigenvectors
form a basis of $\M$, so $G$ is diagonalizable and
\[
  \Tr G=-\tfrac12\cdot2(d-1)-1=-d,\qquad \mmin(G)=-1,\qquad \smax(G)=0 .
\]

\begin{table}[ht]
\centering
\renewcommand{\arraystretch}{1.25}
\begin{tabular}{@{}ccl@{}}
\toprule
Eigenvalue of $G$ & Multiplicity & Eigenvectors\\
\midrule
$0$ & $(d-1)^2$ & $E_{ij}$,\quad $i,j\ne1$\\
$-\tfrac12$ & $2(d-1)$ & $E_{1j},\ E_{j1}$,\quad $j\ne1$\\
$-1$ & $1$ & $E_{11}-E_{00}$\\
\bottomrule
\end{tabular}
\caption{Spectrum of the amplitude-damping generator $G$ of~\eqref{eq:G}.}
\label{tab:spectrum}
\end{table}

A direct check on matrix units shows that
\[
  e^{tG}(X)=K_0(t)XK_0(t)^\dagger+K_1(t)XK_1(t)^\dagger,\qquad
  K_0(t)=I-P+e^{-t/2}P,\quad K_1(t)=\sqrt{1-e^{-t}}\,V .
\]
For example, $e^{tG}(E_{11})=E_{00}+e^{-t}(E_{11}-E_{00})=e^{-t}E_{11}+(1-e^{-t})E_{00}$,
which agrees with the Kraus form. Since
$K_0^\dagger K_0+K_1^\dagger K_1=(I-P+e^{-t}P)+(1-e^{-t})P=I$, the map $e^{tG}$ is
CPTP. Because $\id$ commutes with $G$, we get $e^{tL_c}=e^{ct}e^{tG}$. This map is
completely positive and multiplies traces by $e^{ct}$, which equals $1$ for some $t>0$
only if $c=0$. The spectral data of $L_c$ are those of $G$ shifted by $c$, which gives
the stated values. Finally,
\[
  d\,\mmin(L_c)+(d^2-d)\,\smax(L_c)=d(c-1)+(d^2-d)c=cd^2-d=\Tr L_c .
  \qedhere
\]
\end{proof}

\begin{corollary}\label[corollary]{cor:optimal}
Let $a>d$. For every $d\ge2$ there is a generator $L$ of a non-trace-preserving
completely positive semigroup such that
$\Tr L>a\,\mmin(L)+(d^2-a)\,\smax(L)$.
\end{corollary}

\begin{proof}
For $L=L_c$ with $c\neq0$, we have $\mmin(L)-\smax(L)=-1$, and therefore
$a\,\mmin(L)+(d^2-a)\smax(L)=d^2\smax(L)-a<d^2\smax(L)-d=\Tr L$.
\end{proof}

\section{Decay rates relative to the dominant growth rate}\label{sec:rates}

For non-trace-preserving generators, such as the tilted generators used to study
the full counting statistics of quantum jump trajectories~\cite{CGLP18}, the
leading eigenvalue is generally nonzero. It is then natural to measure decay
relative to it.

\begin{corollary}\label[corollary]{cor:rates}
Let $L$ be conditionally 2-positive with eigenvalues $\lambda_1,\dots,\lambda_{d^2}$,
counted with algebraic multiplicity, and define the relative decay rates
$\Gamma_j=\smax(L)-\Rea\lambda_j\ge0$. Then
\begin{equation}\label{eq:rates}
  \max_{k}\Gamma_k\le\frac1d\sum_{j=1}^{d^2}\Gamma_j ,
\end{equation}
and~\eqref{eq:rates} is equivalent to~\eqref{eq:main}.
\end{corollary}

\begin{proof}
By~\ref{P1}, $\Tr L=\sum_j\Rea\lambda_j$. Hence
$\sum_j\Gamma_j=d^2\smax(L)-\Tr L$ and $\max_k\Gamma_k=\smax(L)-\mmin(L)$.
Inequality~\eqref{eq:rates} therefore reads
$d\,\smax(L)-d\,\mmin(L)\le d^2\smax(L)-\Tr L$, which is a rearrangement
of~\eqref{eq:main}.
\end{proof}

For trace-preserving semigroups, $\smax(L)=0$ and~\eqref{eq:rates} reduces to the
universal constraint~\eqref{eq:rates-classical}. Thus the restriction on relaxation
rates imposed by 2-positivity remains valid without trace preservation.
The relative rates are invariant under scalar shifts $L\mapsto L+c\,\id$
with $c\in\R$, since both $\smax(L)$ and every $\Rea\lambda_j$ shift by $c$.
General non-trace-preserving perturbations need not be scalar shifts and may
change the relative rates; the theorem asserts that the same universal
inequality still holds.
By \cref{prop:sharp}, \eqref{eq:rates} is attained by non-trace-preserving
dynamics.

\section{Concluding remarks}\label{sec:conclusion}

We have shown that the trace-preservation hypothesis in the universal spectral
bound of~\cite{vEC+26} can be removed, both for 2-positive maps and for their
generators, once the dominant eigenvalue is restored as $\smax$. The proof treats
the trace-preserving theorem as a black box and combines it with two soft tools:
a Doob-type normalization by a positive definite eigenmatrix of the adjoint, and a
perturbation that guarantees such an eigenmatrix exists. The same reduction
applies to a continuous, positively homogeneous spectral inequality on an
ambient class of positive maps, provided that class is stable under the
normalization congruences, positive rescaling, and the perturbation
$\Phi\mapsto\Phi+\eps\Delta$, and the inequality is known for its
trace-preserving members. These closure requirements concern the ambient
class, not its trace-preserving subset.

The other questions raised in~\cite{vEC+26} concern whether complete
positivity can be spectrally distinguished from 2-positivity, and whether
positive trace-preserving maps in dimension $d\ge3$ can violate~\eqref{eq:vEC-intro}.
We do not address these questions here. Since normalization and perturbation
use only positivity, validity of the latter bound for all positive
trace-preserving maps would imply its homogeneous extension to all positive
maps; conversely, a trace-preserving counterexample would already be a
counterexample in the larger class.

\section*{Acknowledgements and use of AI tools}

The author thanks the authors of~\cite{vEC+26} for formulating the problem
clearly and for making their work openly available.

In accordance with arXiv policy, AI tools are not listed as authors. OpenAI
ChatGPT was used to assist in preparing this paper. This assistance included
exploring proof strategies, drafting and checking derivations, and editing the
exposition. All mathematical content was reviewed by the author, who takes full
responsibility for the correctness and originality of the work. To the author's
knowledge, the conjecture of~\cite[Sec.~3]{vEC+26} had not been resolved in the
literature at the time of writing. The present argument has not yet undergone
independent peer review.

\section*{Data and Code Availability}
The supporting code package used to verify the analytical claims in this paper is open-source and publicly available on GitHub. It has been permanently archived on Zenodo at \href{https://doi.org/10.5281/zenodo.23226133}{doi:10.5281/zenodo.23226133}.



\begin{thebibliography}{99}

\bibitem{AL07}
R.~Alicki and K.~Lendi,
\emph{Quantum Dynamical Semigroups and Applications}, 2nd ed.,
Lecture Notes in Physics~717, Springer, Berlin, 2007.

\bibitem{BP02}
H.-P. Breuer and F.~Petruccione,
\emph{The Theory of Open Quantum Systems},
Oxford University Press, Oxford, 2002.

\bibitem{CGLP18}
F.~Carollo, J.~P. Garrahan, I.~Lesanovsky, and C.~P\'erez-Espigares,
Making rare events typical in Markovian open quantum systems,
\emph{Phys. Rev. A} \textbf{98} (2018) 010103(R);
\href{https://arxiv.org/abs/1711.10951}{arXiv:1711.10951}.

\bibitem{CFKO21}
D.~Chru\'sci\'nski, R.~Fujii, G.~Kimura, and H.~Ohno,
Constraints for the spectra of generators of quantum dynamical semigroups,
\emph{Linear Algebra Appl.} \textbf{630} (2021) 293--305.

\bibitem{CKKS21}
D.~Chru\'sci\'nski, G.~Kimura, A.~Kossakowski, and Y.~Shishido,
Universal constraint for relaxation rates for quantum dynamical semigroup,
\emph{Phys. Rev. Lett.} \textbf{127} (2021) 050401.

\bibitem{CvEKM25}
D.~Chru\'sci\'nski, F.~vom Ende, G.~Kimura, and P.~Muratore-Ginanneschi,
A universal constraint for relaxation rates for quantum Markov generators:
complete positivity and beyond,
\emph{Rep. Prog. Phys.} \textbf{88} (2025) 097602;
\href{https://doi.org/10.1088/1361-6633/ae075f}{doi:10.1088/1361-6633/ae075f};
\href{https://arxiv.org/abs/2505.24467}{arXiv:2505.24467}.

\bibitem{EHK78}
D.~E. Evans and R.~H{\o}egh-Krohn,
Spectral properties of positive maps on $C^*$-algebras,
\emph{J. London Math. Soc.} (2) \textbf{17} (1978) 345--355.

\bibitem{Eva77}
D.~E. Evans,
Conditionally completely positive maps on operator algebras,
\emph{Q. J. Math.} \textbf{28} (1977) 271--283.

\bibitem{Kat80}
T.~Kato,
\emph{Perturbation Theory for Linear Operators}, 2nd ed.,
Springer, Berlin, 1980.

\bibitem{Kim02}
G.~Kimura,
Restriction on relaxation times derived from the Lindblad-type master equations
for two-level systems,
\emph{Phys. Rev. A} \textbf{66} (2002) 062113.

\bibitem{MKC25}
P.~Muratore-Ginanneschi, G.~Kimura, and D.~Chru\'sci\'nski,
Universal bound on the relaxation rates for quantum Markovian dynamics,
\emph{J. Phys. A: Math. Theor.} \textbf{58} (2025) 045306.

\bibitem{RH12}
\'A.~Rivas and S.~F. Huelga,
\emph{Open Quantum Systems: An Introduction},
Springer, Berlin, Heidelberg, 2012.

\bibitem{vEC+26}
F.~vom Ende, D.~Chru\'sci\'nski, G.~Kimura, and P.~Muratore-Ginanneschi,
Universal bound on the eigenvalues of 2-positive trace-preserving maps,
\emph{Linear Algebra Appl.} \textbf{730} (2026) 262--275;
\href{https://arxiv.org/abs/2506.02145}{arXiv:2506.02145}.

\bibitem{Wat18}
J.~Watrous,
\emph{The Theory of Quantum Information},
Cambridge University Press, Cambridge, 2018.

\bibitem{WECC08}
M.~M. Wolf, J.~Eisert, T.~S. Cubitt, and J.~I. Cirac,
Assessing non-Markovian quantum dynamics,
\emph{Phys. Rev. Lett.} \textbf{101} (2008) 150402.

\end{thebibliography}
\end{document}